\documentclass[conference]{IEEEtran}
\pdfoutput=1
 
\usepackage{booktabs}
\usepackage[utf8]{inputenc}
\usepackage{subcaption}
\usepackage[font={small,bf},labelfont=bf,format=plain]{caption}
\usepackage{textcomp}
\usepackage{amssymb,amsmath,amsthm}
\usepackage{mathrsfs}

\usepackage[inline]{enumitem}
\usepackage{algorithm}
\usepackage{algpseudocode}
\algtext*{EndProcedure}
\algtext*{EndFor}
\algtext*{EndWhile}
\algtext*{EndIf}
\usepackage{xcolor,colortbl}
\usepackage{footnote}
\usepackage{url}
\usepackage{hyperref}
\usepackage{tabularx}
\usepackage{multicol}
\usepackage{multirow}
\usepackage{pifont}
\usepackage{bm}
\usepackage{listings}
\usepackage{wrapfig}
\usepackage{relsize}
\usepackage{mathtools}
\usepackage{bbm}
\usepackage{float}
\usepackage{blindtext}
\usepackage{balance}
\usepackage{hyphenat}

\usepackage{wasysym}
\usepackage{amssymb}
\usepackage{pifont}

\newcommand{\gray}[1]{\textcolor{gray}{#1}}
\newcommand{\mymathhl}[1]{\colorbox{gray!15}{$#1$}}

\newcommand{\method}{\textsc{LinkWaldo}\xspace}
\newcommand{\methodDeg}{\method-D\xspace}
\newcommand{\methodMulti}{\method-M\xspace}
\newcommand{\netmf}{\textsc{NetMF}\xspace}
\newcommand{\netmfOne}{\textsc{NetMF-1}\xspace}
\newcommand{\netmfTwo}{\textsc{NetMF-2}\xspace}
\newcommand{\bine}{\textsc{BiNE}\xspace}
\newcommand{\xnetmf}{\textsc{xNetMF}\xspace}
\newcommand{\groupingOne}{DG\xspace}
\newcommand{\groupingTwo}{SG\xspace}
\newcommand{\groupingThree}{CG\xspace}
\newcommand{\groupingFour}{MG\xspace}

\newcommand{\expectedVal}{\mathbb{E}}
\newcommand{\var}{Var}

\newcommand{\graph}{\mathcal{G}}
\newcommand{\numNodes}{n}
\newcommand{\numEdges}{m}
\newcommand{\adj}{\mathbf{A}}
\newcommand{\adjEl}{a}

\newcommand{\edges}{\mathcal{E}}

\newcommand{\newEdges}{\edges_{\text{new}}}

\newcommand{\nodes}{\mathcal{V}}
\newcommand{\neighbors}{\mathcal{N}}

\newcommand{\pairs}{\mathcal{P}}
\newcommand{\globalPool}{\tilde{\pairs}_G}
\newcommand{\budget}{k}
\newcommand{\toler}{\tau}
\newcommand{\bailoutTol}{\zeta}
\newcommand{\target}{\kappa} 
\newcommand{\targetMean}{\bar{\mu}}
\newcommand{\targetVar}{\sigma}

\newcommand{\numHashes}{b}
\newcommand{\numHashesMax}{\numHashes_{\max}}
\newcommand{\numTrees}{r}

\newcommand{\buckets}{\mathcal{B}}
\newcommand{\bucketVolume}{\text{Vol}}
\newcommand{\bucket}{\beta}

\newcommand{\hash}{h}
\newcommand{\concatHash}{g}
\newcommand{\rhHashFamily}{\mathcal{H}_{\text{rh}}}
\newcommand{\andHashFamily}{\mathcal{H}_{\text{and}}}
\newcommand{\randomHype}{\mathbf{r}_h}

\newcommand{\kl}{D_{\text{KL}}}
\newcommand{\totalErr}{\xi}
\newcommand{\testDist}{p_{\text{n}}}
\newcommand{\trainDist}{p_{\text{o}}}
\newcommand{\varDist}{d_{TV}}

\newcommand{\ourModel}{FLLM\xspace}
\newcommand{\sbm}{SBM\xspace}
\newcommand{\lpm}{LaPM\xspace}
\newcommand{\prm}{PM\xspace}
\newcommand{\simi}{sim} 

\newcommand{\mem}{\boldsymbol{\mu}}
\newcommand{\memele}{\mu}
\newcommand{\ind}{I}
\newcommand{\memER}{\sim_\ind}
\newcommand{\compat}{{\rho}} 
\newcommand{\group}{\mathcal{V}}
\newcommand{\grouping}{\Gamma}

\newcommand{\groupingWeight}{\mathbf{W}}
\newcommand{\groupingWeightEl}{w}
\newcommand{\memberPart}{\Pi}  
\newcommand{\memberPartLong}{\memberPart = \{\mathcal{C}_1, \allowbreak \mathcal{C}_2, \dots, \allowbreak \mathcal{C}_{|\memberPart|}\}}
\newcommand{\memberPartEl}{\mathcal{C}_i} 
\newcommand{\memberPartElArb}{\mathcal{C}} 
\newcommand{\decompX}{\mathcal{V}_u}
\newcommand{\decompY}{\mathcal{V}_v}

\newcommand{\dimension}{d}
\newcommand{\emb}{\mathbf{x}} 
\newcommand{\embs}{\mathbf{X}}

\newcommand{\reals}{\mathbb{R}}
\newcommand{\realsd}{\reals^\dimension}
\newcommand{\naturals}{\mathbb{N}}

\newcommand{\order}{O}

\newcommand{\pr}{\text{Pr}}

\newcommand{\complexityGroupings}{\gamma}

\newcommand{\sOne}{\textbf{S1}\xspace}
\newcommand{\sTwo}{\textbf{S2}\xspace}
\newcommand{\sThree}{\textbf{S3}\xspace}
\newcommand{\sFour}{\textbf{S4}\xspace}

\newtheorem{theorem}{Theorem}

\newtheorem{problem}{Problem}
\newtheorem{definition}{Definition}

\newtheorem{fact}{Fact}
\newtheorem{example}{Example}

\newcommand{\rOne}{\textbf{RQ1}}
\newcommand{\rTwo}{\textbf{RQ2}}
\newcommand{\rThree}{\textbf{RQ3}}

\newcommand{\cn}{\textsc{CN}\xspace}
\newcommand{\js}{\textsc{JS}\xspace}
\newcommand{\adad}{\textsc{AA}\xspace}
\newcommand{\ensemble}{\textsc{NMF+Bag}\xspace}

\newcommand{\yeast}{Yeast\xspace}
\newcommand{\dblp}{DBLP\xspace}
\newcommand{\arxiv}{arXiv\xspace}
\newcommand{\fbOne}{Facebook1\xspace}
\newcommand{\fbTwo}{Facebook2\xspace}
\newcommand{\movie}{MovieLens\xspace}
\newcommand{\protein}{HS-Protein\xspace}
\newcommand{\proteinAlt}{Protein-Soy\xspace}
\newcommand{\mathOverflow}{MathOverflow\xspace}
\newcommand{\enron}{Enron\xspace}
\newcommand{\reddit}{Reddit\xspace}
\newcommand{\epinions}{Epinions\xspace}
\newcommand{\digg}{Digg\xspace}

\newcommand{\yeastCite}{\yeast \cite{zhang2018link}\xspace}
\newcommand{\dblpCite}{\dblp \cite{kunegis2013konect}\xspace}
\newcommand{\arxivCite}{\arxiv \cite{snapnets}\xspace}
\newcommand{\fbOneCite}{\fbOne \cite{snapnets}\xspace}
\newcommand{\fbTwoCite}{\fbTwo \cite{kunegis2013konect}\xspace}
\newcommand{\movieCite}{\movie \cite{kunegis2013konect}\xspace}
\newcommand{\proteinCite}{\protein \cite{kunegis2013konect}\xspace}
\newcommand{\proteinAltCite}{\proteinAlt \cite{snapnets}\xspace}
\newcommand{\mathOverflowCite}{\mathOverflow \cite{snapnets}\xspace}
\newcommand{\enronCite}{\enron \cite{kunegis2013konect}\xspace}
\newcommand{\redditCite}{\reddit \cite{snapnets}\xspace}
\newcommand{\epinionsCite}{\epinions \cite{kunegis2013konect}\xspace}
\newcommand{\diggCite}{\digg \cite{rossi2015network}\xspace}

\begin{document}

\title{A Hidden Challenge of Link Prediction: \\ Which Pairs to Check?}

\author{
\IEEEauthorblockN{Caleb Belth}
\IEEEauthorblockA{\emph{University of Michigan} \\
Ann Arbor, MI, USA \\
cbelth@umich.edu}
\and
\IEEEauthorblockN{Alican B\"uy\"uk\c{c}ak{\i}r}
\IEEEauthorblockA{\emph{University of Michigan} \\ 
Ann Arbor, MI, USA \\ 
alicanb@umich.edu}
\and
\IEEEauthorblockN{Danai Koutra}
\IEEEauthorblockA{\emph{University of Michigan} \\ 
Ann Arbor, MI, USA \\ 
dkoutra@umich.edu}
}

\thispagestyle{plain}
\pagestyle{plain}

\maketitle

\begin{abstract}
The traditional setup of link prediction in networks assumes that a test set of node pairs, which is usually balanced, is available over which to predict the presence of links. However, in practice, there is no test set: the ground-truth is not known, so the number of possible pairs to predict over is quadratic in the number of nodes in the graph. Moreover, because graphs are sparse, most of these possible pairs will not be links. Thus, link prediction methods, which often rely on proximity-preserving embeddings or heuristic notions of node similarity, face a vast search space, with many pairs that are in close proximity, but that should not be linked. To mitigate this issue, we introduce \method, a framework for choosing from this quadratic, massively-skewed search space of node pairs, a concise set of candidate pairs that, in addition to being in close proximity, also \emph{structurally resemble the observed edges}. This allows it to ignore some high-proximity but low-resemblance pairs, and also identify high-resemblance, lower-proximity pairs. Our framework is built on a model that theoretically combines Stochastic Block Models (SBMs) with node proximity models. The block structure of the SBM maps out \emph{where} in the search space new links are expected to fall, and the proximity identifies the most plausible links within these blocks, using locality sensitive hashing to avoid expensive exhaustive search. \method can use any node representation learning or heuristic definition of proximity, and can generate candidate pairs for any link prediction method, allowing the representation power of current and future methods to be realized for link prediction \emph{in practice}. We evaluate \method on 13 networks across multiple domains, and show that on average it returns candidate sets containing 7-33\% more missing and future links than both embedding-based and heuristic baselines' sets.
\end{abstract}

\section{Introduction}
\label{sec:intro}

Link prediction is a long-studied problem that attempts to predict either missing links in an incomplete graph, or links that are likely to form in the future. This has applications in discovering unknown protein interactions to speed up the  discovery of new drugs, friend recommendation in social networks, knowledge graph completion, and more \cite{adamic2003friends, liben2007link, martinez2016survey,safavi2020evaluating}. Techniques range from heuristics, such as predicting links based on the number of common neighbors between a pair of nodes, to machine learning techniques,
which formulate the link prediction problem as a binary classification problem over node pairs \cite{hamilton2017representation, zhang2018link}. 

\begin{figure}[t]
    \centering
    \includegraphics[width=\columnwidth, trim=0 0 2.5cm 0, clip]{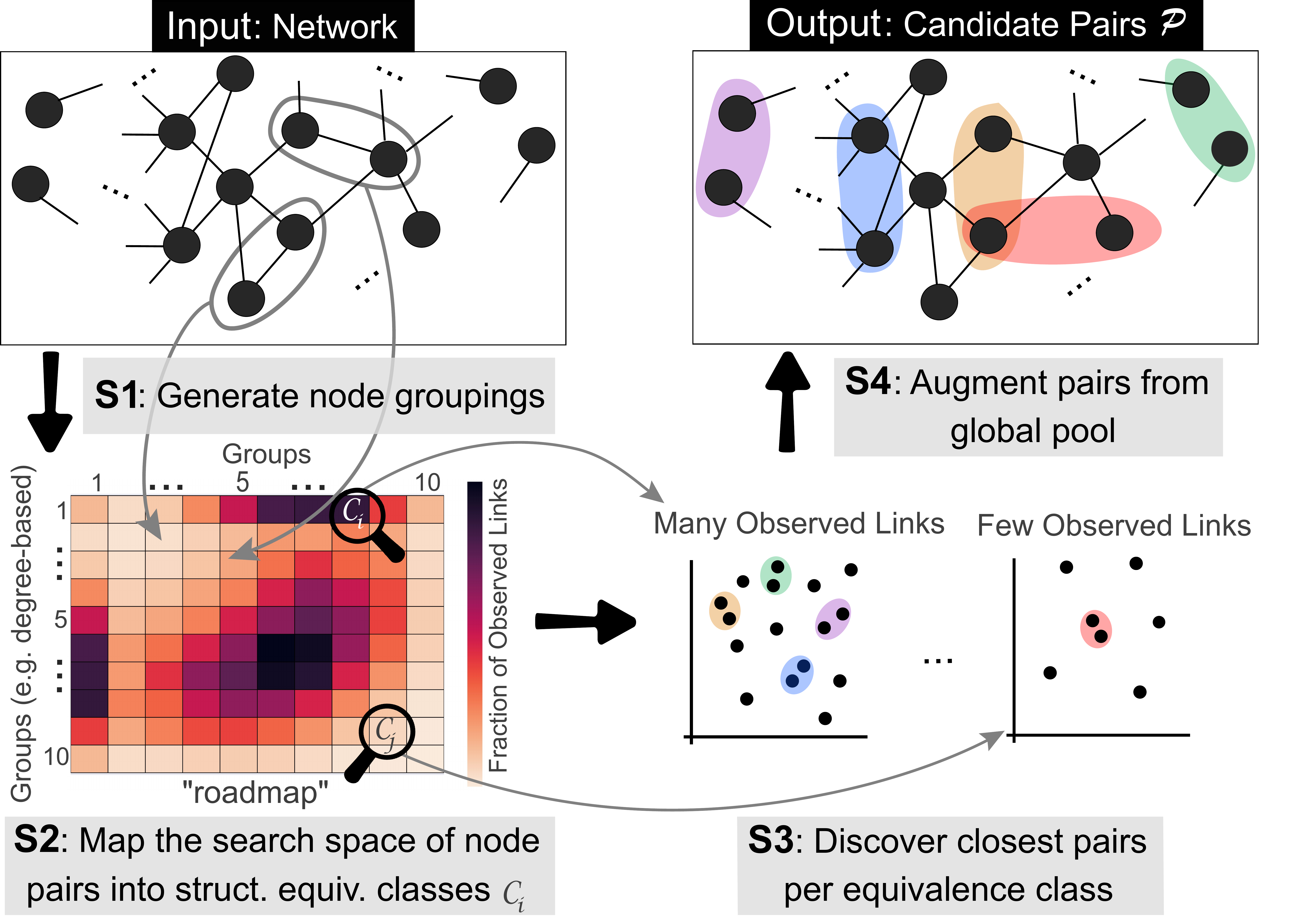}
    \caption{Our proposed framework \method chooses candidate pairs from the quadratic, highly-skewed search space of possible links by first constructing a \emph{roadmap}, which partitions the search space into structural equivalence classes of node pairs to capture how much pairs in each location resemble the \emph{observed} links. This roadmap tells \method how closely to look in each section of the search space. \method follows the roadmap, selecting from each equivalence class the node pairs in closest proximity.} 
    \label{fig:crown-jewel}
    \vspace{-0.45cm}
\end{figure}

Link prediction is often evaluated via a \emph{ranking}, where pairs of nodes that are not currently linked are sorted based on the ``likelihood'' score given by the method being evaluated \cite{martinez2016survey}. To construct the ranking, a ``ground-truth'' test set of node pairs is constructed by either (1) removing a certain percentage of links from a graph at random or (2) removing the newest links that formed in the graph, if edges have timestamps. These removed edges form the test positives, and the same number of unlinked pairs are generated at random as test negatives. The methods are then evaluated on how well they are able to rank the test positives higher than the test negatives. 

However, when link prediction is applied in practice, these ground truth labels are not known, since that is the very question that link prediction is attempting to answer. Instead, \emph{any} pair of nodes that are not currently linked could link in the future. Thus, to identify likely missing or future links, a link prediction method would need to consider $\order(n^2)$ node pairs for a graph with $\numNodes$ nodes; most of which in sparse, real-world networks would turn out to not link. Proximity, on its own, is only a weak signal, sufficient to rank pairs in a balanced test set, but likely to turn up many false positives in an asymptotically skewed space, leaving discovering the relatively small number of missing or future links a challenging problem. 

Proximity-based link prediction heuristics \cite{liben2007link}, such as Common Neighbors, could ignore some of the search space, such as nodes that are farther than two hops from each other, but this would not extend to other notions of proximity, like proximity-preserving embeddings.
Duan \emph{et al.} studied the problem of pruning the search space \cite{duan2017ensemble}, but formulated it as top-$k$ link prediction, which attempts to predict a small number of links, but misses a large number of missing links in the process, suffering from low recall. 

The goal of this work is to develop a principled approach to choose, from the quadratic and skewed space of possible links, a set of candidate pairs for a link prediction method to make decisions about.
We envision that this will allow current and future developments 
to be realized for link prediction in practice, where no ground-truth set is available.

\begin{problem}
Given a graph and a proximity function between nodes, we seek to return a candidate set of node pairs for a link predictor to make decisions about, such that the set is significantly smaller than the quadratic search space, but contains many of the missing and future links. 
\end{problem}

Our insight to handle the vast number of negatives is to consider not \emph{just} the proximity of nodes, but also their \textit{structural resemblance} to observed links. We measure resemblance as the fraction of observed links that fall in inferred, graph-structural equivalence classes of node pairs. For example, Fig.~\ref{fig:crown-jewel} shows one possible grouping of nodes based on their degrees, where the resulting structural equivalence classes (the cells in the ``roadmap'') capture what fraction of observed links form between nodes of different degrees. Based on the roadmap, equivalence classes with a high fraction of observed edges are expected to contain more unlinked pairs than those with lower resemblance. We then employ node proximity within equivalence classes, rather than globally, which decreases false positives that are in close proximity, but do not resemble observed links, and decreases false negatives that are farther away in the graph, but resemble many observed edges. Moreover, to avoid computing proximities for all pairs of nodes within each equivalence class, we extend self-tuning locality sensitive hashing (LSH).
Our main \textbf{contributions} are:
\begin{itemize}
    \item \textbf{Formulation \& Theoretical Connections.} Going beyond the heuristic of proximity between nodes, we model the plausibility of a node pair being linked as \emph{both} their proximity \emph{and} their structural resemblance to observed links. Based on this insight, we propose Future Link Location Models (\ourModel), which combine Proximity Models and Stochastic Block Models; and we prove that Proximity Models are a naive special case. 
    \S~\ref{sec:theory}
    \item \textbf{Scalable Method.} We develop a scalable method, \method (Fig.~\ref{fig:crown-jewel}), which implements \ourModel, and uses locality sensitive hashing to implicitly ignore unimportant pairs. \S~\ref{sec:method}
    \item \textbf{Empirical Analysis.} We evaluate \method 
    on 13 diverse datasets from different domains, where 
    it returns on average 22-33\% more missing links than embedding-based models and 7-30\% more than strong heuristics.
    \S~\ref{sec:evaluation}
\end{itemize}
Our code is at {\href{https://github.com/GemsLab/LinkWaldo}{https://github.com/GemsLab/LinkWaldo}}.

\section{Related work}
In this paper, we focus on the understudied problem of choosing candidate pairs from the quadratic space of possible links, for link prediction methods to make predictions about. Link prediction techniques range from heuristic definitions of similarity, such as Common Neighbors \cite{liben2007link}, Jaccard Similarity \cite{liben2007link}, and Adamic-Adar \cite{adamic2003friends}, to machine learning approaches, such as latent methods, which learn low-dimensional node representations that preserve graph-structural proximity in latent space \cite{hamilton2017representation}, and GNN methods, which \emph{learn} heuristics specific to each graph \cite{zhang2018link} or attempt to re-construct the observed adjacency matrix \cite{kipf2016variational}. For detailed discussion of link prediction techniques, we refer readers to \cite{liben2007link} and \cite{martinez2016survey}.

\vspace{0.1cm}
\noindent \textit{Selecting Candidate Pairs.}  
The closest problem to ours is top-$\budget$ link prediction \cite{duan2017ensemble}, which attempts to take a particular link prediction method and prune its search space to directly return the $\budget$ highest score pairs. 
One method \cite{duan2017ensemble} samples multiple subgraphs to form a bagging ensemble, and 
performs NMF on each subgraph, returning the nodes with the largest latent factor products from each, while leveraging early-stopping. 
The authors view their method's output as predictions rather than candidates, and thus focus on high precision at small values of $\budget$ relative to our setting. 
Another approach, Approximate Resistance Distance Link Predictor~\cite{pachev2018fast} generates spectral node embeddings by constructing a low-rank approximation of the graph's \textit{effective resistance} matrix, and applies a $\budget$-closest pairs algorithm on the embeddings, predicting these as links. 
However, this approach does not scale to moderate embedding dimensions (e.g., the dimensionality of 128 often-used used in embedding methods), and is often outperformed by the simple common neighbors heuristic.

A related problem is link-recommendation, which seeks to identify the $\budget$ most relevant nodes to a query node. It has been studied in social networks for friend recommendation \cite{song2015top}, and in knowledge graphs \cite{joshi2020searching} to pick subgraphs that are likely to contain links to a given query entity. In contrast, we focus on candidate pairs globally, not specific to a query node.

\section{Theory}
\label{sec:theory}
\label{subsec:prelims}

\begin{table}[b!]
\vspace{-0.2cm}
\caption{Description of major symbols.}
\label{tab:Symbols}
\vspace{-0.2cm}
\centering
\resizebox{\columnwidth}{!}{
  \begin{tabular}{ll}
  \toprule
     \textbf{Notation} & \textbf{Description}  \\ \midrule
     $\graph = (\nodes, \edges), \adj$ & Graph, nodes, edges, adjacency matrix \\
     $|\nodes|=\numNodes, |\edges|=\numEdges$ & Number of nodes resp. edges in $\graph$\\
     $\newEdges$ & Unobserved future or missing links\\
     $\grouping$, $\memberPart$ & Grouping of $\nodes$, Partition of $\nodes \times \nodes$ \\
     $\emb_v \in \embs, \mem_v$ & Node embedding and membership vector \\
     $\memberPartEl$ & Equivalence class $i$\\
     $\pairs$, $\globalPool$ & Pairs selected by \method, global pool\\
     $\budget, \target$ & Budget for $|\pairs|$, target for an equiv. class\\
     \bottomrule
    \end{tabular}
    }
\end{table}

Let $\graph = (\nodes, \edges)$ be a graph or network with $|\nodes|=n$ nodes and $|\edges|=m$ edges, where $\edges \subseteq \nodes \times \nodes$. 
The adjacency matrix $\adj$ of $\graph$ is an $n \times n$ binary matrix with element $\adjEl_{ij} = 1$ if nodes $i$ and $j$ are linked, and 0 otherwise. The set of node $v$'s neighbors 
is $\neighbors(v)=\{u: (u,v)\in \edges\}$. We summarize the key symbols used in this paper and their descriptions in Table~\ref{tab:Symbols}.

We now formalize the problem that we seek to solve:
\begin{problem}
\label{problem:prob-description}
Given a graph $\graph = (\nodes, \edges)$, a proximity function $\simi: \nodes \times \nodes \rightarrow \reals^+$ between nodes, and a budget $\budget << \numNodes^2$, return a set of plausible candidate node pairs $\pairs \subset \nodes \times \nodes$ of size $|\pairs| = \budget$ for a link predictor to make decisions about.
\end{problem}

We describe next how to define resemblance in a principled way inspired by Stochastic Block Models, introduce a unified model for link prediction methods that use the proximity of nodes to rank pairs, and describe our model, which combines resemblance and proximity to solve Problem~\ref{problem:prob-description}.

\subsection{Stochastic Block Models}

Stochastic Block Models (SBMs) are generative models of networks. They model the connectivity of graphs as emerging from the community or group membership of nodes \cite{nowicki2001estimation}. 

\vspace{0.1cm}
\noindent \textit{Node Grouping.} A node grouping $\grouping$ is a set of groups or subsets $\group_i$ of the nodes that satisfies $\bigcup_{\group_i \in \grouping} \group_i = \nodes$. It is called a \emph{partition} if it satisfies $\group_i \cap \group_j = \emptyset \;\; \forall \group_i \neq \group_j \in \grouping$. Each node $v \in \nodes$ has a $|\grouping|$-dimensional binary membership vector $\mem_v$, with element $\memele_{vi} = 1$ if $v$ belongs to group $\group_i$. 

A node grouping can capture community structure, but it can also capture other graph-structural properties, like the degrees of nodes, in which case the SBM captures the compatibility of nodes w.r.t degree---viz. degree assortativity.

\vspace{0.1cm}
\noindent \textit{Membership Indices}. The membership indices $\ind_{u, v}$ of nodes $u, v$ are the set of group ids $(i, j)$ s.t. $u\in\group_i$ and $v\in\group_j$: 
$\ind_{u, v} \triangleq 
\{i : \memele_{u,i} = 1\} \times \{j : \memele_{v,j} = 1\}$, \; $i, j \in \{ 1, 2, \dots, |\grouping| \}$. 

\vspace{0.1cm}
\noindent \textit{Membership equivalence relation \& classes}.
The membership indices form the equivalence relation $\memER$: 
$(u, v) \memER (u', v') \iff \ind_{u, v} = \ind_{u', v'}$. 
This induces a partition $\memberPartLong$ over all pairs of nodes $\nodes \times \nodes$ (both linked and unlinked), where the equivalence class $\memberPartEl$ contains all node pairs $(u,v)$ with the same membership indices, i.e., $\mem_u = \mem$ and $\mem_v = \mem'$ for some $\mem, \mem' \in \{0, 1\}^{|\grouping|}$.  
We denote the equivalence class of pair $(u, v)$ as $[(u, v)]_{\memER}$.

\begin{example}
\label{ex:equiv-class-0} If nodes are grouped by their degrees to form $\grouping$, then the membership indices $\ind_{u, v}$ of node pair $(u, v)$ are determined by $u$ and $v$'s respective degrees. For example, in Fig.~\ref{fig:crown-jewel}, the upper circled node pair has degrees $3$ and $5$ respectively, which determines their equivalence class---in this case, the cell $(3, 5)$ in the roadmap. Each cell of the roadmap corresponds to an equivalence class $\memberPartElArb_i \in \memberPart$.

\end{example}

We can now formally define an SBM:
\begin{definition}[Stochastic Block Model - \sbm]
Given a node grouping $\grouping$ and a $|\grouping| \times |\grouping|$ weight matrix $\groupingWeight$ specifying the propensity for links to form across groups, 
the probability that two nodes link given their group memberships is $\mymathhl{\pr(\adjEl_{uv} = 1|\mem_u, \mem_v) = \sigma(\mem_u^T \groupingWeight \mem_v)}$, where function $\sigma(\cdot)$ converts the dot product to a probability (e.g., sigmoid)~\cite{miller2009nonparametric}.
\end{definition}

The vanilla SBM \cite{nowicki2001estimation} assigns each node to one group (i.e., the grouping is a partition and membership vectors $\mem$ are one-hot), in which case $\mem_u^T \groupingWeight \mem_v = \groupingWeightEl_{\ind_{u,v}}$. The overlapping SBM \cite{miller2009nonparametric, latouche2011overlapping} 
is a generalization that allows nodes to belong to multiple groups, in which case membership vectors may have multiple elements set to 1, and $\mem_u^T \groupingWeight \mem_v = \sum_{i,j \in \ind_{u, v}} \groupingWeightEl_{ij}$.

\vspace{0.1cm}
\noindent \textit{Resemblance.} Given an SBM with grouping $\grouping$, we define the resemblance of node pair $(u, v) \in \nodes \times \nodes$ under the SBM as the percentage of the observed (training) edges that have the same group membership as $(u,v)$:
\begin{equation}
\label{eq:compat}
    \compat(u, v) \triangleq \frac{|\{(v_1, v_2) \in \edges : (v_1, v_2) \memER (u, v) \}|}{\numEdges}.
\end{equation}

\begin{example}
\label{ex:equiv-class} 
In Figure~\ref{fig:crown-jewel}, the resemblance $\compat(u, v)$ of node pair $(u,v)$ corresponds to the density of the cell that it maps to. 
The high density in the border cells indicates that 
many low-degree nodes connect to high-degree nodes. 
The dense central cells indicate that mid-degree nodes connect to each other.
\end{example}

\subsection{Proximity Models}
\label{subsec:prox}

Proximity-based link prediction models (\prm) model the connectivity of graphs based on the proximity of nodes. Some methods define the proximity of nodes with a heuristic, such as Common Neighbors (\cn), Jaccard Similarity (\js), and Adamic/Adar (\adad). More recent approaches \emph{learn} latent similarities between nodes, capturing the proximity in latent embeddings such that nodes that are in close proximity in the graph have similar latent embeddings (e.g., dot product) \cite{hamilton2017representation}.

\vspace{0.1cm}
\noindent \textit{Node Embedding.} A node embedding, $\emb_v \in \realsd$, is a real-valued, $\dimension$-dimensional vector representation of a node $v \in \nodes$. 
We denote all the node embeddings as matrix $\embs \in \reals^{n \times \dimension}$.

\begin{definition}[Proximity Model - \prm] {Given a similarity or proximity function $\simi : \nodes \times \nodes \rightarrow \reals^{+}$ between nodes, the probability that nodes $u$ and $v$ link
is an increasing function of their proximity: 
$\mymathhl{\pr(\adjEl_{uv} = 1|\simi(\cdot,\cdot)) = f(\simi(u, v))}$.}
\end{definition}

Instances of the \prm include the Latent Proximity Model:

\begin{equation}
\label{eq:simi-lapm}
    \simi_{\text{\lpm}}(u, v) \triangleq \emb_u^T\emb_v,
\end{equation}
where $\emb_u, \emb_v$ are the nodes' latent embeddings; and the 
Common Neighbors, Jaccard Similarity, and Adamic/Adar models:
\begin{equation}
\label{eq:simi-cn}
    \simi_{\cn}(u, v) \triangleq |\neighbors(u) \cap \neighbors(v)|,
\end{equation}
\begin{equation}
\label{eq:simi-js}
    \simi_{\js}(u, v) \triangleq \frac{|\neighbors(u) \cap \neighbors(v)|}{|\neighbors(u) \cup \neighbors(v)|}, \text{ and}
\end{equation}
\begin{equation}
\label{eq:simi-aa}
    \simi_{\adad}(u, v) \triangleq \sum_{v' \in \neighbors(u) \cap \neighbors(v)}\frac{1}{\log|\neighbors(v')|}.
\end{equation}

\subsection{Proposed: Future Link Location Model}

Unlike \sbm and \prm, our model, which we call the \emph{Future Link Location Model} (\ourModel), is not just modeling the probability of links, but rather \emph{where} in the search space future links are likely to fall. 
To do so, \ourModel
uses a partition of the search space, and corresponding SBM, as a roadmap that
gives the number of new edges expected to fall in each equivalence class. To formalize this idea, we first define two distributions:

\vspace{0.1cm}
\noindent \textit{New and Observed Distributions}. 
The new link distribution $\testDist(\memberPartEl) \triangleq \pr(\memberPartEl|\newEdges)$ and the observed link distribution $\trainDist(\memberPartEl) \triangleq \pr(\memberPartEl|\edges)$ capture the fraction of new and observed edges that fall in equivalence class $\memberPartEl$, respectively.

\begin{definition}[Future Link Location Model - \ourModel]
\label{dfn:fllm}
Given an overlapping \sbm with grouping $\grouping$, the expected number of new links in equivalence class $\memberPartEl$ is proportional to the number of \emph{observed} links in $\memberPartEl$, and the probability of node pair $(u,v)$ linking is equal to the pair's resemblance times their proximity relative to other nodes in $[(u,v)]_{\memER}$:
\resizebox{\columnwidth}{!}{ $\mymathhl{\pr(\adjEl_{uv} = 1 | \mem_u, \mem_v, \simi(\cdot, \cdot)) = \compat(u,v) \cdot \frac{\simi(u,v)}{\sum_{(u',v') \in [(u, v)]_{\memER}} \simi(u', v')}}$}.
\end{definition}

\ourModel employs the following theorem, which states that if $q\%$ of the observed links fall in equivalence class $\memberPartEl$, then in expectation, $q\%$ of the \emph{unobserved} links will fall in equivalence class $\memberPartEl$. We initially assume that the unobserved future links follow the same distribution as the observed links---as generally assumed in machine learning---i.e., the relative fraction of links in each equivalence class will be the same for future links as observed links: $\testDist = \trainDist$. In the next subsection, we show that for a fixed $\budget$, the error in this assumption is determined by the total variation distance
between $\testDist$ and $\trainDist$, and hence is upper-bounded by a constant.

\begin{theorem}
\label{thm:exp-val}
Given an overlapping SBM with grouping $\grouping$ inducing the partition $\memberPart$ of $\nodes \times \nodes$ for a graph $\graph = (\nodes, \edges)$, out of $\budget$ new (unobserved) links $\newEdges$, the expected number that will fall in equivalence class $\memberPartEl$ and its variance are:
\begin{equation}
    \label{eq:exp-val-def}
    \expectedVal[|\memberPartEl \cap \edges_{new}|] = \frac{\budget|\memberPartEl \cap \edges|}{\numEdges}
\end{equation}
\begin{equation}
    \label{eq:var-def}
    \var(|\memberPartEl \cap \newEdges|) = \frac{\budget|\memberPartEl \cap \edges||\edges \setminus \memberPartEl|}{\numEdges^2}. 
\end{equation}
\end{theorem}

\begin{proof} Observe that the number of the $\budget$ new edges that fall in 
equivalence class $\memberPartEl$, i.e., $|\memberPartEl \cap \newEdges|$,
is a binomial random variable over $\budget$ trials, with success probability 
$\pr(\memberPartEl|\newEdges)$. 
Thus, the random variable's expected value is

{\small
\begin{equation}
    \label{eq:exp-val-1}
    \expectedVal[|\memberPartEl \cap \newEdges|] = \budget \pr(\memberPartEl|\newEdges),
\end{equation}
}
and its variance is 
{\small
\begin{equation}
    \label{eq:var-1}
    \var(|\memberPartEl \cap \newEdges|) = \budget \pr(\memberPartEl|\newEdges) (1 - \pr(\memberPartEl|\newEdges)).
\end{equation}
}
We can derive $\pr(\memberPartEl|\newEdges)$ via $\pr(\memberPartEl|\edges)$ and Bayes' rule:

{\small
\begin{equation}
    \pr(\memberPartEl |\edges) = \frac{\pr(\edges|\memberPartEl) \pr(\memberPartEl)}{\pr(\edges)} 
    = \frac{\frac{|\memberPartEl \cap \edges|}{|\memberPartEl|}\frac{|\memberPartEl|}{|\nodes \times \nodes|}}{\numEdges/|\nodes \times \nodes|}
    = \frac{|\memberPartEl \cap \edges|}{\numEdges}. \nonumber
\end{equation}
}

Combining the last equation with Eq.~\eqref{eq:exp-val-1} results directly in Eq.~\eqref{eq:exp-val-def}, and by substituting into Eq.~\eqref{eq:var-1} we obtain: 

{\small
\begin{equation}
    \label{eq:var-final}
    \var(|\memberPartEl \cap \newEdges|) =  \frac{\budget|\memberPartEl \cap \edges|}{\numEdges} (1 - \frac{|\memberPartEl \cap \edges|}{\numEdges})
    = \frac{\budget|\memberPartEl \cap \edges||\edges \setminus \memberPartEl|}{\numEdges^2},
    \nonumber 
\end{equation}
} 
where we used the fact that $|\edges \setminus \memberPartEl|= |\edges| - |\memberPartEl \cap \edges|$. 
\end{proof}

\subsection{Guarantees on Error}

While this derivation assumed that the future link distribution is the same as the observed link distribution, we now show that for a fixed $\budget$, the amount of error incurred when this assumption does not hold is entirely dependent on the total variation distance, and hence is upper-bounded by $2\budget$.

\noindent \emph{Total Variation Distance.} The total variation distance \cite{tsybakov2008introduction} between $\testDist$ and $\trainDist$, which is a metric, is defined as
\begin{equation}
    \varDist(\testDist, \trainDist) \triangleq \sup_{\mathcal{A} \subset \memberPart} |\testDist(\mathcal{A}) - \trainDist(\mathcal{A})|.
\end{equation} 

\noindent 
\emph{Total Error.} The total error made in the approximation of $\expectedVal[|\memberPartEl \cap \newEdges|]$ using Eq.~\eqref{eq:exp-val-def} is defined as
\begin{align}
    \totalErr &\triangleq \sum_{\memberPartEl \in \memberPart} |\hat{\expectedVal}[|\memberPartEl \cap \newEdges|] - \expectedVal[|\memberPartEl \cap \newEdges|]| \nonumber \\
    &= \sum_{\memberPartEl \in \memberPart}|\budget \trainDist(\memberPartEl) - \budget \testDist(\memberPartEl)|,
    \label{eq:error}
\end{align}
where $\hat{\expectedVal}[|\memberPartEl \cap \newEdges|]$ is the true expected value regardless of whether or not $\testDist = \trainDist$ holds.
\begin{theorem}
\label{thm:err-exp-val}
The total error incurred over $\memberPart$ in the computation of the expected number of new edges that fall in each $\memberPartEl \in \memberPart$ is an increasing function of the the number of new pairs $\budget$ and the total variation distance between $\testDist$ and $\trainDist$. Furthermore, it has the following upper-bound:
\begin{equation}
    \totalErr = 2 \budget \, \varDist(\testDist, \trainDist) \leq \min(2\budget, 2 \budget  \sqrt{1/2\kl(\testDist||\trainDist)}).
\end{equation}
\end{theorem}
\begin{proof} From the definition of total error in Eq.~\eqref{eq:error}, the first equality holds from \cite{levin2017markov}. The inequality holds based on the fact that $\varDist(\cdot, \cdot)$ ranges in $[0, 1]$, and Pinsker's inequality \cite{tsybakov2008introduction}, which upper-bounds $\varDist(\cdot, \cdot)$ via KL-divergence.
\end{proof}

\subsection{Proximity Model as a Special Case of \ourModel}
\label{subsec:pm-special-case}

The \prm, defined in \S~\ref{subsec:prox}, is a special case of \ourModel, where \ourModel's grouping contains just one group $\grouping = \{\nodes\}$. That is, if the nodes are not grouped, then the models give the same result. Thus, \ourModel's improvement over \lpm is a result of using structurally-meaningful groupings over the graph. The following theorem states this result formally.

\begin{theorem}
    {For a single node grouping $\grouping = \{\nodes\}$, both \prm and \ourModel give the same ranking of pairs $(u, v) \in \nodes \times \nodes$}:
    {\footnotesize
    \begin{align}
        \pr_{\text{\prm}}(\adjEl_{uv} = 1|\simi(\cdot, \cdot)) &> \pr_{\text{\prm}}(\adjEl_{u'v'} = 1|\simi(\cdot, \cdot)) \iff \nonumber \\ \pr_{\text{\ourModel}}(\adjEl_{uv} = 1 | \mem_u, \mem_v, \simi(\cdot, \cdot)) &> \pr_{\text{\ourModel}}(\adjEl_{u'v'} = 1 | \mem_{u'}, \mem_{v'}, \simi(\cdot, \cdot)). \nonumber
    \end{align}
    }
\end{theorem}

\begin{proof}
Since $\grouping = \{\nodes\}$, $\memberPart = \{\nodes \times \nodes\}$, and since $\edges \subseteq \nodes \times \nodes$, all observed edges fall in the lone equivalence class $\memberPartElArb = \nodes \times \nodes$. Thus $\compat(u, v) = 1 \ \forall (u, v) \in \nodes \times \nodes$. Since there is only one equivalence class, the denominator in Dfn.~\ref{dfn:fllm} is equal to a constant $c \triangleq \sum_{(u',v') \in [(u,v)]_{\memER}} \simi(u',v') = \sum_{(u',v') \in \nodes \times \nodes} \simi(u', v') \ \forall (u, v) \in \nodes \times \nodes$. Therefore, $\pr_{\text{\ourModel}}(\adjEl_{uv} = 1 | \mem_u, \mem_v, \simi(\cdot, \cdot)) = \frac{1}{c}\simi(u, v)$, {and both models are increasing functions of $\simi(\cdot,\cdot)$.}
\end{proof}

\section{Method}
\label{sec:method}

We solve Problem~\ref{problem:prob-description} by using our FLLM model 
in a new method, 
\method, shown in Fig.~\ref{fig:crown-jewel},  which has four steps:
\begin{itemize}
    \item \sOne: Generate node groupings and equivalence classes.
    \item \sTwo: Map the search space, deciding how many candidate pairs to return from 
    each equivalence class.
    \item \sThree: Search each equivalence class, returning directly the highest-proximity pairs, and stashing some slightly lower-proximity pairs in a global pool.
    \item \sFour: Choose the best pairs from the global pool to augment those returned from each equivalence class.
\end{itemize}

We discuss these steps next, give pseudocode in Alg.~\ref{alg:main-algo}, and discuss time complexity in the appendix.

\subsection{Generating Node Groupings (\sOne)}
\label{subsec:groupings}

In theory, we would like to infer the groupings that directly maximize the likelihood of the observed adjacency matrix. However, the techniques for inferring these groupings (and the corresponding node membership vectors) are computationally intensive, relying on Markov chain Monte Carlo (MCMC) methods \cite{Mehta19sbmgnn}. Indeed, these methods are generally applied in networks with only up to a few hundred nodes \cite{miller2009nonparametric}. In cases where $\numNodes$ is large enough that considering all $\order(\numNodes^2)$ node pairs would be computationally infeasible, so would be MCMC. 
Instead \method uses a fixed grouping, though it is agnostic to how the nodes are grouped. We discuss a number of sensible groupings below, and discuss how to set the number of groups in \S~\ref{subsec:params}. Any other grouping can be readily used within our framework, but should be carefully chosen to lead to strong results.

\vspace{0.1cm}
\noindent $\bullet$ \emph{Log-binned Node Degree (\groupingOne).} This grouping captures degree assortativity \cite{newman2003mixing}, i.e., the extent to which low degree nodes link with other low degree nodes vs. high degree nodes, by creating uniform bins in log-space (e.g., Fig.~\ref{fig:crown-jewel};  linear bins). 

\vspace{0.1cm}
\noindent $\bullet$ \emph{Structural Embedding Clusters (\groupingTwo).} 
This grouping extends \groupingOne by clustering latent node embeddings that capture structural roles of nodes~\cite{Rossi2019FromCT}.  

\vspace{0.1cm}
\noindent $\bullet$ \emph{Communities (\groupingThree).} This grouping captures community structure by clustering proximity preserving latent embeddings or using community detection methods. 

\vspace{0.1cm}
\noindent $\bullet$ \emph{Multiple Groupings (\groupingFour).} Any subset of these groupings or any other groupings can be combined into a new  
grouping, by setting $\mem_v$ element(s) to 1 for $v$'s membership in each grouping,
since nodes can have overlapping group memberships.

\begin{algorithm}[t]
\renewcommand{\algorithmicindent}{0.8em}
	\caption{\small \textsc{\method}($\graph$, $\simi(\cdot, \cdot)$, $\budget, \toler$)}
	\label{alg:main-algo}
	\begin{algorithmic}[1]
	{\small
	   \State \gray{/* \sOne: Generating Node Groupings */}
	   \State Generate node grouping $\grouping$ inducing partition $\memberPart$ \Comment{\gray{\S~\ref{subsec:groupings}}}
	   \State $\pairs, \globalPool \gets \emptyset, \emptyset$ \Comment{\gray{Initialize pairs to return and global pool}}
	   \For{$\memberPartEl \in \memberPart$} \Comment{\gray{Search each equivalence class \S~\ref{subsec:closest-pairs}}}
	       \State \gray{/* \sTwo: Mapping the Search Space */}
	       \State $\targetMean \gets \expectedVal[|\memberPartEl \cap \newEdges|]$ \Comment{\gray{Eq.~\eqref{eq:exp-val-def}}}
	       \State $\targetVar \gets \sqrt{\var(|\memberPartEl \cap \newEdges|)}$ \Comment{\gray{Eq.~\eqref{eq:var-def}}}
	       \State \gray{/* \sThree: Discovering Closest Pairs per Equivalence Class */}
	       \If{$|\memberPartEl| < \toler$} 
	           \State $\Call{SelectPairsExact}{\pairs, \globalPool, \memberPartEl, \targetMean, \targetVar}$
	       \Else
	           \State $\Call{SelectPairsApprox}{\pairs, \globalPool, \memberPartEl, \targetMean, \targetVar}$
	       \EndIf
	   \EndFor
	   \State \gray{/* \sFour: Augmenting Pairs from Global Pool */}
	   \State $\pairs \gets$ \, $\pairs \cup \{\text{top } \budget - |\pairs| \text{ pairs from } \globalPool\}$
	   \State \Return $\pairs$
	\vspace{0.1cm}   
	\Procedure{SelectPairsExact}{$\pairs$, $\globalPool$, $\memberPartEl$, $\targetMean$, $\targetVar$} 
	    \State Sort pairs $(u, v) \in \memberPartEl$ in descending order on $\simi(u, v)$
	    \State $\pairs \gets$ $\pairs \cup$ \{top $\targetMean - \targetVar$ pairs\}
	    \State $\globalPool \gets \globalPool \cup$  \{next $2\targetVar$ pairs\}
	\EndProcedure
	\vspace{0.1cm}
	\Procedure{SelectPairsApprox}{$\pairs$, $\globalPool$, $\memberPartEl$, $\targetMean$, $\targetVar$}
	    \For{$i = 1, 2, \dots, \numTrees$} \Comment{\gray{Create $\numTrees$ trees}}
	       \State $\buckets \gets \{(\decompX, \decompY)\}$ \Comment{\gray{buckets start off as the root}}
    	    \While{$\bucketVolume(\buckets) > \target$} \Comment{\gray{cf. Prob.~\ref{prob:closest-pairs} for $\target$ definition}}
    	       \State Choose $\hash(\cdot)$ at random from $\rhHashFamily$
    	       \State $\buckets' \gets \emptyset$ \Comment{\gray{Create new buckets}}
    	       \For{$\bucket \in \buckets$} \Comment{\gray{Branch each leaf (bucket)}}
	               \State $\bucket_{\text{left}}' \gets$ \resizebox{0.72\columnwidth}{!}{$(\{u \in \decompX^{(\bucket)} : \hash(\emb_u) < 0 \}, \{v \in \decompY^{(\bucket)} : \hash(\emb_v) < 0 \})$}
	               \State $\bucket_{\text{right}}' \gets$ \resizebox{0.72\columnwidth}{!}{$(\{u \in \decompX^{(\bucket)} : \hash(\emb_u) \geq 0 \}, \{v \in \decompY^{(\bucket)} : \hash(\emb_v) \geq 0 \})$}
	               \State $\buckets' \gets \buckets' \cup \{\bucket_{\text{left}}, \bucket_{\text{right}}\}$
	           \EndFor
	           \If{$\bucketVolume(\buckets) \leq \target$} 
	                $\buckets \gets \buckets'$
	           \EndIf
    	    \EndWhile
    	\EndFor
	    \State $\pairs \gets$ $\pairs \cup$ \{top $\targetMean - \targetVar$ pairs\}
	    \State $\globalPool \gets \globalPool \cup$  \{next $2\targetVar$ pairs\}
	\EndProcedure
	}
	\end{algorithmic}
	\vspace{-0.1cm}
\end{algorithm}

\subsection{Mapping the Search Space (\sTwo)}
\label{subsec:mapping-search-space}
\method's approach to mapping the search space (i.e., identifying how many pairs to return per class $\memberPartEl$) follows directly from Thm.~\ref{thm:exp-val}. \method computes the expected number of pairs in each equivalence class based on Eq.~\eqref{eq:exp-val-def} 
and its variance based on Eq.~\eqref{eq:var-def}, as a measure of the uncertainty. When \method searches each equivalence class $\memberPartEl$, 
it returns the expected number of pairs \emph{minus} a standard deviation \emph{directly}, and adds more pairs, up to 
a standard deviation past the mean,
to a global pool $\globalPool$. Thus, \method adds into $\pairs$ the $\expectedVal[|\memberPartEl \cap \edges_{new}|] - \sqrt{\var(|\memberPartEl \cap \newEdges|)}$ pairs in closest proximity in equivalence class $\memberPartEl$, and the next $2\sqrt{\var(|\memberPartEl \cap \newEdges|)}$ closest pairs into the global pool $\globalPool$ (both expressions are rounded to the nearest integer). Node pairs that are already linked are skipped. 
\subsection{Discovering Closest Pairs per Equivalence Class (\sThree)}
\label{subsec:closest-pairs}
We now discuss how \method discovers the $\target$ closest unlinked pairs \emph{within} each equivalence class (Fig.~\ref{fig:crown-jewel}), where $\target$ is determined in step \sTwo based on the expected number of pairs in the equivalence class, and variance (uncertainty).

\begin{problem}
\label{prob:closest-pairs}
Given an equivalence class $\memberPartEl$, return the top-$\target$ unlinked pairs in $\memberPartEl$ in closest proximity $\simi(\cdot, \cdot)$, 
where $\target=\expectedVal[|\memberPartEl \cap \edges_{new}|] + \sqrt{\var(|\memberPartEl \cap \newEdges|)}$ (based on \sTwo). 
\end{problem}

For equivalence classes smaller than some tolerance $\toler$, it is feasible to search all pairs of nodes exhaustively. However, for $|\memberPartEl|>\toler$, this should be avoided, to make the search practical. We first discuss this case when using the dot product similarity $\simi_{\text{\lpm}}(\cdot, \cdot)$ in Eq.~\eqref{eq:simi-lapm}, and then discuss it for other similarity models (CN, JS, and AA) given by Eqs.~\eqref{eq:simi-cn}-\eqref{eq:simi-aa}. Finally, we introduce a refinement that improves the robustness of \method against errors in proximity.

\subsubsection{Avoiding Exhaustive Search for Dot Product} In the case of dot product, we 
use Locality Sensitive Hashing (LSH) \cite{wang2014hashing} to avoid searching all $|\memberPartEl|$ pairs. 
LSH functions have the property that the probability of two items colliding is a function of their similarity. We use the following fact:
\begin{fact}
\label{fact:decomp}
The equivalence class $\memberPartEl$ 
can be decomposed into the Cartesian product of two sets $\memberPartEl = \decompX \times \decompY$, where $\decompX \triangleq \{u : \mem_u = \mem \}$ and $\decompY = \{v : \mem_v = \mem'\}$. 
\end{fact}
At a high level, to solve Prob.~\ref{prob:closest-pairs}, we hash each node embedding of the nodes in $\decompX$ and $\decompY$ using a locality sensitive hash function. We design the hash function, described next, such that the number of pairs that map to the same bucket is greater than $\target$, but as small as possible, to maximally prune pairs. 
Once the embeddings are hashed, we search the pairs in each hash bucket for the $\target$ closest.
We normalize the embeddings so that dot product is equivalent to cosine similarity, and use the Random Hyperplane LSH family \cite{charikar2002similarity}.

\begin{definition}[Random Hyperplane Hash Family] The random hyperplane hash family $\rhHashFamily$ is the set of hash functions $\rhHashFamily \triangleq \{\hash : \realsd \rightarrow \{0, 1\}\}$, where $\randomHype$ is a random $\dimension$-dimensional Gaussian unit vector and $\hash(\emb) \triangleq \begin{cases} 1 & \text{if } \randomHype^T \emb \geq 0 \\ 0 & \text{if } \randomHype^T \emb < 0 \end{cases}$. 
\end{definition}

This hash family is well-known to provide the property that the probability of two vectors colliding is a function of the degree of the angle between them \cite{bawa2005lsh}:
\begin{equation}
\small
    \pr(\hash(\emb_u) = \hash(\emb_v)) = 1 - \frac{\theta(\emb_u,\emb_v)}{\pi} = 1 - \frac{\arccos(\emb_u^T\emb_v)}{\pi} \nonumber,
\end{equation} \vspace{-0.02cm}
\noindent where the last equality holds due to normalized embeddings. 

To lower the false positive rate, it is conventional to form a new hash function by sampling $\numHashes$ hash functions from $\rhHashFamily$ and concatenating the hash codes: $\concatHash(\cdot) = (\hash_1(\cdot), \hash_2(\cdot), \dots, \hash_\numHashes(\cdot))$. The new hash function is from another LSH family:

\begin{definition}[$\numHashes$-AND-Random Hyperplane Hash Family] The $\numHashes$-AND-Random hyperplane hash family is the set of hash functions $\andHashFamily^{\numHashes} \triangleq \{\concatHash : \realsd \rightarrow \{0, 1\}^\numHashes\}$, where $\concatHash(\emb) = (\hash_1(\emb), \hash_2(\emb), \dots, \hash_\numHashes(\emb))$ is formed by concatenating $\numHashes$ randomly sampled hash functions $\hash(\cdot) \in \rhHashFamily$ for some $\numHashes \in \naturals$.
\end{definition}

Since the hash functions are sampled randomly from $\rhHashFamily$, 
\begin{equation}
\label{eq:concat-prob}
\small 
    \pr(\concatHash(\emb_u) = \concatHash(\emb_v)) = \left(1 - \frac{\arccos(\emb_u^T\emb_v)}{\pi}\right)^\numHashes. \nonumber
\end{equation}

Only vectors that are not split by all $\numHashes$ random hyperplanes end up with the same hash codes, so this process lowers the false positive rate. However, it also increases the false negative rate for the same reason. The conventional LSH-scheme then repeats the process $\numTrees$ times, computing the dot product exactly over all pairs that match in at least one $\numHashes$-dim hash code, in order to lower the false negative rate. The challenge of this approach is determining how to set $\numHashes$. To do so, we first define the hash buckets of a hash function, and their volume. 
\begin{definition}[Hash Buckets and Volume] 
Given an equivalence class $\memberPartEl = \decompX \times \decompY$ and a hash function \vspace{0.5cm} $\concatHash(\cdot) : \realsd \rightarrow \{0, 1\}^\numHashes$, after applying $\concatHash(\cdot)$ to all $v \in \decompX \cup \decompY$, a hash bucket 
$$\bucket = \{ u \in \decompX, v \in \decompY : \concatHash(u) = \concatHash(v) =  \bucket_{\text{hashcode}}\}$$
consists of subsets $\decompX^{(\bucket)} \subseteq \decompX, \decompY^{(\bucket)} \subseteq \decompY$ of nodes that mapped to hashcode $\bucket_{\text{hashcode}} \in \{0, 1\}^\numHashes$. The set of hash buckets $\buckets_\concatHash = \{\bucket : |\bucket| > 0\}$ consists of all non-empty buckets. 
We define the volume of the buckets as the number of pairs $(u, v)$ where $u$ and $v$ landed in the same bucket:
{\small
\begin{equation}
\small 
    \bucketVolume(\buckets_\concatHash) \triangleq |\{(u, v) : \concatHash(\emb_u) = \concatHash(\emb_v)\}| = \textstyle \sum_{\bucket \in \buckets_\concatHash} |\decompX^{(\bucket)} \times \decompY^{(\bucket)}|. \nonumber
\end{equation}
}
\end{definition}

Since we are after the $\target$ closest pairs, we want to find a hash function $\concatHash(\cdot)$ such that $\bucketVolume(\buckets_\concatHash) \geq \target$.
But since we want to search as few pairs as possible, we seek the value of $\numHashes$ that minimizes $\bucketVolume(\buckets_\concatHash)$ for some $\concatHash(\cdot) \in \andHashFamily^{\numHashes}$ subject to the constraint that $\bucketVolume(\buckets_\concatHash) \geq \target$.

Any hash function $\concatHash \in \andHashFamily^{\numHashes}$ corresponds to a binary prefix tree, like Fig.~\ref{fig:lsh-tree}. 
Each level of the tree corresponds to one
\begin{wrapfigure}{r}{0.35\linewidth}
\vspace{-0.4cm}
    \centering
    \includegraphics[width=\linewidth]{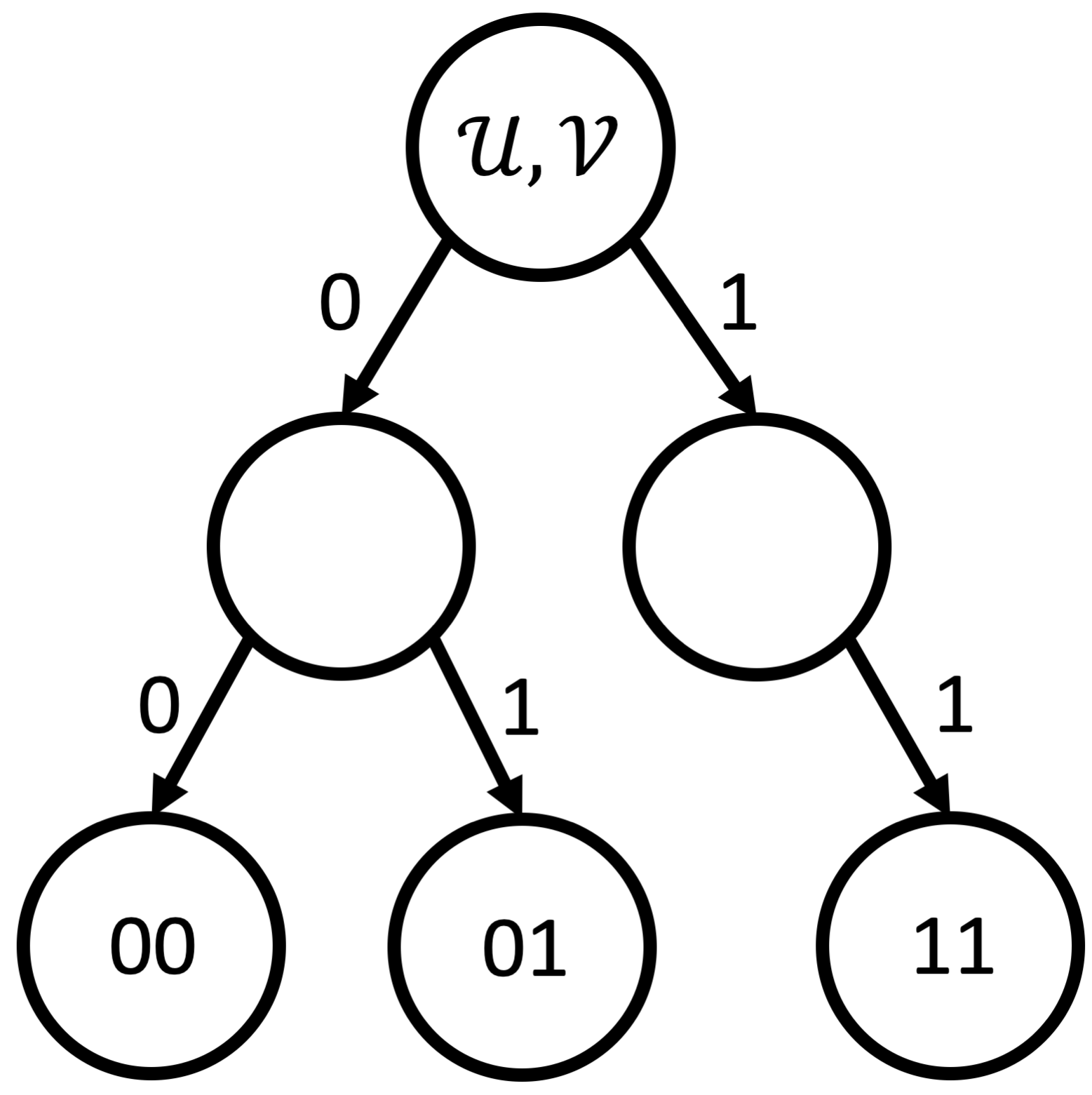}
    \vspace{-0.5cm}
    \caption{LSH Tree}
    \label{fig:lsh-tree}
    \vspace{-0.5cm}
\end{wrapfigure}
$\hash \in \rhHashFamily$, and the leaves correspond to the buckets $\buckets_\concatHash$. Thus, to automatically identify the best value of $\numHashes$, we can recursively grow the tree, branching each leaf with a new random hyperplane hash function $\hash \in \rhHashFamily$, until $\bucketVolume(\buckets_\concatHash) < \target$, then undo the last branch. At that point, the depth of the tree equals $\numHashes$, and is the largest value such that $\bucketVolume(\buckets_\concatHash) \geq \target$. To prevent this process from repeating indefinitely in edge cases, we halt the branching at a maximum depth $\numHashesMax$. This approach is closely related to LSH Forests \cite{bawa2005lsh}, but with some key differences, which we discuss below.

\begin{theorem}
Given a hash function $\concatHash \in \andHashFamily^{\numHashes}$, the $\target$ closest pairs in $\memberPartEl$ are the $\target$ most likely pairs to be in the same bucket: 
\resizebox{\columnwidth}{!}{$\pr(\concatHash(\emb_u) = \concatHash(\emb_v)) > \pr(\concatHash(\emb_{u'}) = \concatHash(\emb_{v'})) \iff \emb_u^T\emb_v > \emb_{u'}^T\emb_{v'}$}.
\end{theorem}
\begin{proof} Since $\arccos(x)$ is a \emph{decreasing} function of $x$, Eq.~\eqref{eq:concat-prob} shows that $\pr(\concatHash(\emb_u) = \concatHash(\emb_v))$ is an \emph{increasing} function of $\emb_u^T\emb_v$. The result follows from this. 
\end{proof}
While $\emb_u^T\emb_v > \emb_{u'}^T\emb_{v'}$ implies that $(u,v)$ are \emph{more likely} than $(u',v')$ to be in the same bucket, it does not guarantee that this outcome will \emph{always} happen. Thus, we repeat the process $\numTrees$ times, creating $\numTrees$ binary prefix trees and, searching the pairs that fall in the same bucket in any tree for the top $\target$. Setting the $\numTrees$ parameter is considered of minor importance, as long as it is sufficiently large (e.g., 10) \cite{bawa2005lsh}.

\vspace{0.1cm}
\noindent \emph{Differences from LSH Forests \cite{bawa2005lsh}}. LSH Forests are designed for $KNN$-search, which seeks to return the nearest neighbors to a \emph{query vector}. In contrast, our approach is designed for $\target$-closest-pairs search, which seeks to return the $\target$ closest pairs in a set $\memberPartEl$. 
LSH Forests grow each tree until each vector is in its own leaf. 
We grow each tree until we reach the target bucket volume $\target$. LSH Forests allow variable length hash codes, since the nearest neighbors of different query vectors may be at different relative distances. All our leaves are at the same depth so that the probability of $(u, v)$ surviving together to the leaf is an increasing function of their dot product.  

\subsubsection{Avoiding Exhaustive Search for Heuristics}
\label{subsubsec:heuristic-avoid-all} 
For the heuristic definitions of proximity in Eqs.\eqref{eq:simi-cn}-\eqref{eq:simi-aa}, there are two approaches to solving Prob.~\ref{prob:closest-pairs}. The first is to construct embeddings from the \cn and \adad scores (this does not apply to \js). For \cn, if we let the node embeddings be their corresponding rows in the adjacency matrix, i.e, $\embs_{\cn} = \adj$, then $\simi_{\cn}(u, v) = \emb_u^T\emb_v$. Similarly, $\embs_{\adad} = \adj \cdot 1/\sqrt{\log(\mathbf{D})}$, yields $\simi_{\adad}(u, v) = \emb_u^T\emb_v$, where $\mathbf{D}$ is a diagonal matrix recording the degree of each node. Thus, the LSH solution just described can be applied. The second approach uses the fact that all three heuristics are defined over the 1-hop neighborhoods of nodes $(u, v)$. Thus, to have nonzero proximity, $(u, v)$ must be within 2-hops of each other, and any pairs not within 2-hops can implicitly be ignored.
\label{subsubsec:bailout}

\subsubsection{Bail Out Refinement}
To this point we have assumed that the proximity model used in \method is highly informative and accurate. However, in reality, heuristics may not be informative for all equivalence classes, and even learned, latent proximity models, can fail to encode adequate information. For instance, it is challenging to learn high-quality representations for low-degree nodes. 
Thus, we introduce a refinement to \method that automatically identifies when a proximity model is uninformative in an equivalence class, and allows it to \emph{bail out} of searching that equivalence class.

\vspace{0.1cm}
\noindent \emph{Proximity Model Error.} The error that a proximity model makes is the probability $\pr(\simi(u, v) < \simi(u', v'))$ that it gives a higher proximity for some unlinked pair $(u', v') \notin \edges$ than for some linked pair $(u, v) \in \edges$.

By this definition of error, we expect strong proximity models to mostly assign higher proximity between observed edges than future or missing edges: $\pr(\simi(u, v) > \simi(u', v')) \approx 1$ for some $(u, v) \in \edges$ and $(u', v') \in \newEdges$. Thus, on our way to finding the top-$\target$ most similar (unlinked) pairs in an equivalence class (Problem~\ref{prob:closest-pairs}), we expect to encounter a majority of the \emph{observed} edges (linked pairs) $|\edges \cap \memberPartEl|$ that fall in that class. For a user-specified error tolerance $\bailoutTol$, \method will bail out and return no pairs from any equivalence class where less than $\bailoutTol$ fraction of its observed edges are encountered on the way to finding the $\target$ most similar unlinked pairs. 
\method keeps track of how many pairs were skipped by bailing out, and replaces them (after step \sFour) by adding to $\pairs$ the top-ranked pairs of a heuristic (e.g., \adad).

\subsection{Augmenting Pairs from Global Pool (\sFour)}
\label{subsec:global}
Since \method returns a standard deviation below the expected number of new pairs in each equivalence class, it chooses the remaining pairs up to $\budget$ from $\globalPool$. To do so, it considers pairs in descending order on the input similarity function $\simi(\cdot, \cdot)$, and greedily adds to $\pairs$ until $|\pairs| = \budget$.

\section{Evaluation}
\label{sec:evaluation}
We evaluate \method on three research questions: 
(\rOne)~Does the set $\pairs$ returned by \method have high recall and precision? 
(\rTwo)~Is \method scalable?
(\rThree)~How do parameters affect performance?

\begin{table}[t]
    \centering
    \caption{Datasets statistics: if the graph is temporal or static, density, degree assortativity \cite{newman2003mixing}, and number of nodes and edges.}
    \label{tab:datasets}
    \resizebox{\columnwidth}{!}{
    \begin{tabular}{lccccc}
         \toprule
         Graph & Time & Density & Assortativity & $n$ & $m$\\
         \midrule
         \yeast & - & 0.41\% & 0.4539 & 2,375 & 11,693 \\
         \dblp & - & 0.06\% & -0.0458 & 12,595 & 49,638 \\
         \fbOne & - & 1.08\% & 0.0636 & 4,041 & 88,235 \\
         \movie & \checkmark & 2.90\% & -0.2268 & 2,627 & 100,000 \\
         \protein & - & 0.74\% & 0.2483 & 6,329 & 147,548 \\
         \arxiv & - & 0.11\% & 0.2051 & 18,772 & 198,110 \\
         \mathOverflow & \checkmark & 0.06\% & -0.1979 & 24,820 & 199,974 \\
         \enron & \checkmark & 0.01\% & -0.1667 & 87,275 & 299,221 \\
         \reddit & \checkmark & 0.01\% & -0.1278 & 67,180 & 309,667 \\
         \epinions & - & 0.01\% & -0.0406 & 75,881 & 405,741 \\
         \fbTwo & \checkmark & 0.04\% & 0.1770 & 63,733 & 817,063\\
         \digg & \checkmark & $< 0.01\%$ & -0.0557 & 279,376 & 1,546,541 \\
         \proteinAlt & - & 1.64\% & -0.0192 & 45,116 & 16,691,679 \\
         \bottomrule
    \end{tabular}
    }
    \vspace{-0.4cm}
\end{table}

\subsection{Data \& Setup}
\label{sec:data}
We evaluate \method on a large, diverse set of networks: metabolic, social, communication, and information networks. Moreover, we include datasets to evaluate in both LP scenarios: (1) returning possible \emph{missing} links in static graphs and (2) returning possible \emph{future} links in temporal graphs. We treat all graphs as undirected. 

\vspace{0.1cm}
\noindent \emph{Metabolic.} \yeastCite, \proteinCite, and \proteinAltCite are metabolic protein networks, where edges denote known associations between proteins in different species. \yeast contains proteins in a species of yeast, \protein in human beings, and \proteinAlt in Glycine max (soybeans).

\vspace{0.1cm}
\noindent \emph{Social.} \fbOneCite and \fbTwoCite capture friendships on Facebook, \redditCite encodes links between subreddits (topical discussion boards), edges in \epinionsCite connect users who trust each other's opinions, \mathOverflowCite captures comments and answers on math-related questions and comments (e.g., user $u$ answered user $v$'s question), \diggCite captures friendships among users.

\vspace{0.1cm}
\noindent \emph{Communication.} \enronCite is an email network, capturing emails sent during the collapse of the Enron energy company.

\vspace{0.1cm}
\noindent \emph{Information.} \dblpCite is a citation network, and \arxivCite is a co-authorship network of Astrophysicists. \movieCite is bipartite graph of users rating movies for the research project \movie. Edges encode users and the movies that they rated.

\vspace{0.1cm}
\noindent \emph{Training Graph and Ground Truth.} While using \method in practice does not require a test set, in order to know how effective it is, we must \emph{evaluate} it on ground truth missing links. As ground truth, we remove 20\% of the edges. In the static graphs, we remove 20\% at random. In the temporal graphs, we remove the 20\% of edges with the most recent timestamps. If either of the nodes in the removed edge is not present in the training graph, we discard the edge from the ground-truth. The graph with these edges removed is the training graph, which \method and the baselines observe when choosing the set of unlinked pairs to return.

\vspace{0.1cm}
\noindent \emph{Setup.} We discuss in \S~\ref{subsec:params} how we choose which groupings to use, and how many groups in each.
Whenever used, we implement \groupingTwo and \groupingThree by clustering embeddings with KMeans: \xnetmf \cite{heimann2018regal} and NetMF \cite{qiu2018network} (window size 1), respectively.
In LSH, we set the maximum tree depth dynamically based on the size of an equivalence class: $\numHashesMax = 12$ if $|\memberPartEl| < 1B$, $\numHashesMax = 15$ if $|\memberPartEl| < 10B$, $\numHashesMax = 20$ $|\memberPartEl| < 25B$, $\numHashesMax = 30$ otherwise. We set the number of trees $\numTrees$ based on the fraction of $|\memberPartEl|$ that we seek to return: $\numTrees = 5$ if $\target / |\memberPartEl| < 0.0001$, $\numTrees = 10$ if $\target / |\memberPartEl| < 0.001$ and $\numTrees = 25$ otherwise.
\subsection{Recall and Precision (\rOne)}

\noindent \textbf{Task Setup.} We evaluate how effectively \method returns in $\pairs$ the ground-truth missing links, 
at values of $\budget$ much smaller than $\numNodes^2$. We report $\budget$, chosen
based on dataset size, in Tab.~\ref{table:recall}, and discuss effects of the choice in the appendix. We compare the set \method returns to those of five baselines, and evaluate both \methodDeg, which uses grouping \groupingOne, and \methodMulti, which uses \groupingOne, \groupingTwo, and \groupingThree together. In both \method variants, we consider the following proximities (cf.~\ref{subsec:prox}) as input, and report the results that are best: \lpm using \netmf \cite{qiu2018network} embeddings (window sizes 1 and 2), and \adad, the 
best 
heuristic proximity. For the bipartite \movie, we use \bine \cite{gao2018bine}, an embedding method designed for bipartite graphs. We report the input proximity model for each dataset in Tab.~\ref{tab:params} in the appendix. We set the exact-search and bailout tolerances to $\toler = 25M$ and $\bailoutTol = 0.5$, which we determined via a parameter study in \S~\ref{subsec:params}. Results are averages over five random seeds (\S~\ref{sec:data}): for static graphs, the randomly-removed edges are different for each seed; for temporal graphs, the latest edges are always removed, so the LSH hash functions are the main source of randomness.

\vspace{0.1cm}
\noindent \textbf{Metrics.} We use Recall (R@$\budget$), the fraction of known missing/future links that are in the size-$\budget$ set returned by the method, and Precision (P@$\budget$), the fraction of the $\budget$ pairs that are known to be missing/future links. Recall is a more important metric, since (1) the returned set of pairs $\pairs$ does not contain final predictions, but rather pairs for a LP method to make final decisions about, and (2) our real-world graphs are inherently incomplete, and thus pairs returned that are not \emph{known} to be missing links, could nonetheless be missing in the original dataset prior to ground-truth removal (i.e., the open-world assumption~\cite{safavi2020evaluating}).  
We report both in Table~\ref{table:recall}.

\vspace{0.1cm}
\noindent \textbf{Baselines.} We use five baselines.  
\textbf{\ensemble}~\cite{duan2017ensemble} uses non-negative matrix factorization (NMF) and a bagging ensemble to return $\budget$ pairs while pruning the search space. We use their reported strongest version: the \emph{Biased Edge Bagging} version with \emph{Node Uptake} and \emph{Edge Filter} optimizations (\textit{Biased(NMF+)}). We use the authors' recommended parameters when possible: $\epsilon = 1$, $\mu = 0.1$, $f = 0.1$, $\rho = 0.75$, number of latent factors $d = 50$, and ensemble size $\mu/f^2$. In some cases, these suggested parameters led to fewer than $\budget$ pairs being returned, in which case we tweaked the values of $\epsilon, \mu$, and $f$ until $\budget$ were returned. We report these deviations in Tab.~\ref{tab:params} in the appendix. We use our own implementation.

We also use four proximity models, which we showed to be special cases of \ourModel in \S~\ref{subsec:pm-special-case}: 
\noindent \textbf{\lpm} ranks pairs \emph{globally} based on the dot product of their embeddings, and returns the top $\budget$. To avoid searching all-pairs, we use the same LSH scheme that we introduce in \S~\ref{subsec:closest-pairs} for \method. We set $\numTrees = 25$, and like \method, use \netmf with a window size of 1 or 2, except for \movie, where we use \bine.

\noindent \textbf{\js}, \textbf{\cn}, and \textbf{\adad} are defined in \ref{subsec:prox}. We exploit the property described in~\ref{subsubsec:heuristic-avoid-all}---i.e., all these scores are zero for nodes beyond two hops. We compute the scores for all nodes within two hops, and return the top $\budget$ unlinked pairs.

\begin{table*}[t]
        \centering
    	\caption{On each dataset, we highlight the cell of the top-performing method with bold text and a gray background, and the second best with bold text only. An ``*'' denotes statistical significance at a 0.05 $p$-value in a paired t-test. The ``**'' means that the difference between the better performing variant of \method was also significantly better than the other variant at the same $p$-value. Under each dataset, we give the percentage of the quadratic search space that the value of $\budget$ corresponds to (usually $< 1\%$). On average, \methodMulti is the best-performing method, and \methodDeg the second best.}
    	\label{table:recall}
    	\resizebox{\textwidth}{!}{
        \begin{tabular}{clcccccccc}
    		\toprule
    		\emph{Dataset} & \emph{Metric} & \ensemble \cite{duan2017ensemble} & \lpm & \js & \cn & \adad & \methodDeg & \methodMulti \\
    		\toprule
    		\yeast & R@10K & 0.4078 $\pm$ 0.01 & 0.4400 $\pm$ 0.01 & 0.4766 $\pm$ 0.01 & 0.6142 $\pm$ 0.01 & 0.6590 $\pm$ 0.01 & \textbf{0.6762 $\pm$ 0.01} & \cellcolor{gray!25}\textbf{0.6926** $\pm$ 0.01} \\
            0.39\% & P@10K & 0.0898 $\pm$ 0.00 & 0.0969 $\pm$ 0.00 & 0.1049 $\pm$ 0.00 & 0.1352 $\pm$ 0.00 & 0.1451 $\pm$ 0.00 & \textbf{0.1489 $\pm$ 0.00} & \cellcolor{gray!25}\textbf{0.1525** $\pm$ 0.00} \\
            \cmidrule{1-9}
            \dblp & R@100K & 0.2319 $\pm$ 0.00 & 0.0927 $\pm$ 0.00 & 0.0389 $\pm$ 0.00 & 0.3379 $\pm$ 0.00 & 0.3775 $\pm$ 0.00 & \textbf{0.4270 $\pm$ 0.00} & \cellcolor{gray!25}\textbf{0.4271* $\pm$ 0.00} \\
            0.15\% & P@100K & 0.0204 $\pm$ 0.00 & 0.0082 $\pm$ 0.00 & 0.0034 $\pm$ 0.00 & 0.0298 $\pm$ 0.00 & 0.0332 $\pm$ 0.00 & \cellcolor{gray!25}\textbf{0.0376* $\pm$ 0.00} & \cellcolor{gray!25}\textbf{0.0376* $\pm$ 0.00} \\
            \cmidrule{1-9}
            \fbOne & R@100K & 0.4036 $\pm$ 0.02 & 0.8005 $\pm$ 0.00 & 0.8244 $\pm$ 0.00 & 0.8547 $\pm$ 0.00 & 0.8863 $\pm$ 0.00 & \textbf{0.8975 $\pm$ 0.00} & \cellcolor{gray!25}\textbf{0.9059** $\pm$ 0.00} \\
            1.34\% & P@100K & 0.0711 $\pm$ 0.00 & 0.1410 $\pm$ 0.00 & 0.1453 $\pm$ 0.00 & 0.1506 $\pm$ 0.00 & 0.1562 $\pm$ 0.00 & \textbf{0.1581 $\pm$ 0.00} & \cellcolor{gray!25}\textbf{0.1596** $\pm$ 0.00} \\
            \cmidrule{1-9}
            \movie & R@100K & 0.1221 $\pm$ 0.02 & 0.2096 $\pm$ 0.01 & 0.0000 $\pm$ 0.00 & 0.0000 $\pm$ 0.00 & 0.0000 $\pm$ 0.00 & \textbf{0.1667 $\pm$ 0.01} & \cellcolor{gray!25}\textbf{0.3662** $\pm$ 0.01} \\
            3.57\% & P@100K & 0.0035 $\pm$ 0.00 & 0.0060 $\pm$ 0.00 & 0.0000 $\pm$ 0.00 & 0.0000 $\pm$ 0.00 & 0.0000 $\pm$ 0.00 & \textbf{0.0048 $\pm$ 0.00} & \cellcolor{gray!25}\textbf{0.0105** $\pm$ 0.00} \\
            \cmidrule{1-9}
            \protein & R@100K & 0.5127 $\pm$ 0.02 & 0.4033 $\pm$ 0.01 & 0.4768 $\pm$ 0.00 & 0.7998 $\pm$ 0.00 & 0.8429 $\pm$ 0.00 & \textbf{0.8878 $\pm$ 0.00} & \cellcolor{gray!25}\textbf{0.9038** $\pm$ 0.00} \\
            0.52\% & P@100K & 0.1506 $\pm$ 0.00 & 0.1185 $\pm$ 0.00 & 0.1400 $\pm$ 0.00 & 0.2349 $\pm$ 0.00 & 0.2476 $\pm$ 0.00 & \textbf{0.2608 $\pm$ 0.00} & \cellcolor{gray!25}\textbf{0.2655** $\pm$ 0.00} \\
            \cmidrule{1-9}
            \arxiv & R@100K & 0.2877 $\pm$ 0.00 & 0.2584 $\pm$ 0.00 & 0.5149 $\pm$ 0.00 & 0.5576 $\pm$ 0.00 & 0.6539 $\pm$ 0.00 & \textbf{0.7004 $\pm$ 0.00} & \cellcolor{gray!25}\textbf{0.7032** $\pm$ 0.00} \\
            0.06\% & P@100K & 0.2017 $\pm$ 0.00 & 0.1812 $\pm$ 0.00 & 0.3610 $\pm$ 0.00 & 0.3909 $\pm$ 0.00 & 0.4585 $\pm$ 0.00 & \textbf{0.4911 $\pm$ 0.00} &  \cellcolor{gray!25}\textbf{0.4930** $\pm$ 0.00} \\
            \cmidrule{1-9}
            \mathOverflow & R@1M & 0.3901 $\pm$ 0.00 & 0.0279 $\pm$ 0.00 & 0.0084 $\pm$ 0.00 & 0.4201 $\pm$ 0.00 & \cellcolor{gray!25}\textbf{0.4333* $\pm$ 0.00} &\textbf{ 0.4227 $\pm$ 0.00} & 0.4233 $\pm$ 0.00 \\
            0.55\% & P@1M & 0.0070 $\pm$ 0.00 & 0.0005 $\pm$ 0.00 & 0.0001 $\pm$ 0.00 & 0.0075 $\pm$ 0.00 & \cellcolor{gray!25}\textbf{0.0078* $\pm$ 0.00} & \textbf{0.0076 $\pm$ 0.00} & \textbf{0.0076 $\pm$ 0.00}\\
            \cmidrule{1-9}
            \enron & R@1M & 0.2551 $\pm$ 0.00 & 0.0008 $\pm$ 0.00 & 0.0004 $\pm$ 0.00 & 0.3071 $\pm$ 0.00 & \cellcolor{gray!25}\textbf{0.3209* $\pm$ 0.00} & \textbf{0.3166 $\pm$ 0.00} & 0.3080 $\pm$ 0.00 \\
            0.04\% & P@1M & 0.0089 $\pm$ 0.00 & 0.0000 $\pm$ 0.00 & 0.0000 $\pm$ 0.00 & 0.0107 $\pm$ 0.00 & \cellcolor{gray!25}\textbf{0.0112* $\pm$ 0.00} & \textbf{0.0111 $\pm$ 0.00} & 0.0108 $\pm$ 0.00 \\
            \cmidrule{1-9}
            \reddit & R@1M & 0.3165 $\pm$ 0.00 & 0.0015 $\pm$ 0.00 & 0.0006 $\pm$ 0.00 & 0.3796 $\pm$ 0.00 & 0.4038 $\pm$ 0.00 & \textbf{0.4096 $\pm$ 0.00} &  \cellcolor{gray!25}\textbf{0.4186** $\pm$ 0.00} \\
            0.06\% & P@1M & 0.0130 $\pm$ 0.00 & 0.0001 $\pm$ 0.00 & 0.0000 $\pm$ 0.00 & 0.0156 $\pm$ 0.00 & 0.0166 $\pm$ 0.00 & \textbf{0.0168 $\pm$ 0.00} & \cellcolor{gray!25}\cellcolor{gray!25}\textbf{0.0172** $\pm$ 0.00} \\
            \cmidrule{1-9}
            \epinions & R@1M & 0.3312 $\pm$ 0.00 & 0.0371 $\pm$ 0.00 & 0.0383 $\pm$ 0.00 & 0.3871 $\pm$ 0.00 & 0.4291 $\pm$ 0.00 & \textbf{0.4292 $\pm$ 0.00} & \cellcolor{gray!25}\textbf{0.4323** $\pm$ 0.00} \\
            0.04\% & P@1M & 0.0242 $\pm$ 0.00 & 0.0027 $\pm$ 0.00 & 0.0028 $\pm$ 0.00 & 0.0283 $\pm$ 0.00 & \textbf{0.0313 $\pm$ 0.00} & \textbf{0.0313 $\pm$ 0.00} & \cellcolor{gray!25}\textbf{0.0316** $\pm$ 0.00} \\
            \cmidrule{1-9}
            \fbTwo & R@1M & 0.0948 $\pm$ 0.00 & 0.1947 $\pm$ 0.00 & 0.3605 $\pm$ 0.00 & 0.2439 $\pm$ 0.00 & 0.2832 $\pm$ 0.00 & \cellcolor{gray!25}\textbf{0.3796**} $\pm$ 0.00 & \textbf{0.3776 $\pm$ 0.00} \\
            0.06\% & P@1M & 0.0145 $\pm$ 0.00 & 0.0298 $\pm$ 0.00 & 0.0551 $\pm$ 0.00 & 0.0373 $\pm$ 0.00 & 0.0433 $\pm$ 0.00 & \cellcolor{gray!25}\textbf{0.0580** $\pm$ 0.00} & \textbf{0.0577 $\pm$ 0.00} \\
            \cmidrule{1-9}
            \digg & R@10M & 0.2459 $\pm$ 0.00 & 0.0035 $\pm$ 0.00 & 0.0015 $\pm$ 0.00 & 0.2952 $\pm$ 0.00 & \cellcolor{gray!25}\textbf{0.3066* $\pm$ 0.00} & \textbf{0.3053 $\pm$ 0.00} & 0.3052 $\pm$ 0.00\\
            0.04\% & P@10M & 0.0032 $\pm$ 0.00 & 0.0000 $\pm$ 0.00 & 0.0000 $\pm$ 0.00 & 0.0038 $\pm$ 0.00 & \cellcolor{gray!25}\textbf{0.0040* $\pm$ 0.00} & \cellcolor{gray!25}\textbf{0.0040* $\pm$ 0.00} & \cellcolor{gray!25}\textbf{0.0040* $\pm$ 0.00} \\
            \cmidrule{1-9}
            \proteinAlt & R@10M & 0.3624 $\pm$ 0.02 & 0.1225 $\pm$ 0.00 & 0.2792 $\pm$ 0.00 & 0.3573 $\pm$ 0.00 & 0.3636 $\pm$ 0.00 & \textbf{0.5781 $\pm$ 0.00} & \cellcolor{gray!25}\textbf{0.6016* $\pm$ 0.03} \\
            0.99\% & P@10M & 0.2178 $\pm$ 0.01 & 0.0736 $\pm$ 0.00 & 0.1678 $\pm$ 0.00 & 0.2147 $\pm$ 0.00 & 0.2185 $\pm$ 0.00 & \textbf{0.3473 $\pm$ 0.00} & \cellcolor{gray!25}\textbf{0.3615* $\pm$ 0.02} \\
            \midrule
            & \centering \textbf{Avg R} & 0.3048 & 0.1994 & 0.2323 & 0.4273 & 0.4585 & \textbf{0.5074} & \cellcolor{gray!25}\textbf{0.5281}\\
            & \centering \textbf{Avg P} & 0.0635 & 0.0506 & 0.0754 & 0.0969 & 0.1056 & \textbf{0.1213} & \cellcolor{gray!25}\textbf{0.1238}\\
      	    \bottomrule
    	\end{tabular}
}
\end{table*}

\vspace{0.1cm}
\noindent \textbf{Results.} Across the 13 datasets, \method is the best performing method on 10, in both recall and precision. The \methodMulti variant is slightly stronger than \methodDeg, but the small gap between the two demonstrates that even simple node groupings can lead to strong improvements over baselines. \method generalizes well across the diverse types of networks. In contrast, the heuristics perform well on social networks, but not as well on, e.g., metabolic networks (\yeast, \protein, and \proteinAlt). Furthermore, the heuristic baselines cannot extend to bipartite graphs like \movie, because fundamentally, all links form between nodes more than one hop away. These observations demonstrate the value of learning from the observed links, which \method does via resemblance. We also observe that heuristic definitions of similarity, such as \adad, outperform latent embeddings (\lpm) that capture proximity. We conjecture that the embedding methods are more sensitive to the massive skew of the data, because even random vectors in high-dimensional space can end up with some level of proximity, due to the curse of dimensionality. This suggests that the standard approach of evaluating on a balanced test set may artificially inflate results.

In the three datasets where \method does not outperform \adad, it is only outperformed by a small margin. Furthermore, the four datasets with the largest total Variation distances between $\testDist$ and $\trainDist$ are \movie, \mathOverflow, \enron, and \digg. Theorem~\ref{thm:err-exp-val} suggests that \method may incur the most error in these datasets. Indeed, these are the only three datasets where \method fails to outperform all other methods (with \movie being bipartite, as discussed above). While the performance on temporal networks is strong, the higher total Variation distance suggests that the assumption that $\trainDist = \testDist$ may sometimes be violated due to \emph{concept drift} \cite{belth2020mining}. Thus, a promising future research direction is to use the timestamps of observed edges to predict roadmap drift over time, in order to more accurately estimate the future roadmap.

\subsection{Scalability (\rTwo)}
\noindent \textbf{Task Setup.} We evaluate how \method scales with the number of edges, and the number of nodes in a graph by running \method with fixed parameters on all datasets. We set $\budget = 1M$, use \netmf (window-size of 1) as $\simi(\cdot, \cdot)$, and do not perform bailout ($\bailoutTol = 0$). All other parameters are identical to $\rOne$. We use our Python implementation on an Intel(R) Xeon(R) CPU E5-2697 v3, 2.60GHz with 1TB RAM.

\vspace{0.1cm}
\noindent \textbf{Results.} The results in Fig.~\ref{fig:runtime} demonstrate that in practice, \method scales linearly on the number of edges, and sub-quadratically on the number of nodes.

\begin{figure}[t]
\vspace{-0.25cm}
    \centering
    \includegraphics[width=.9\columnwidth]{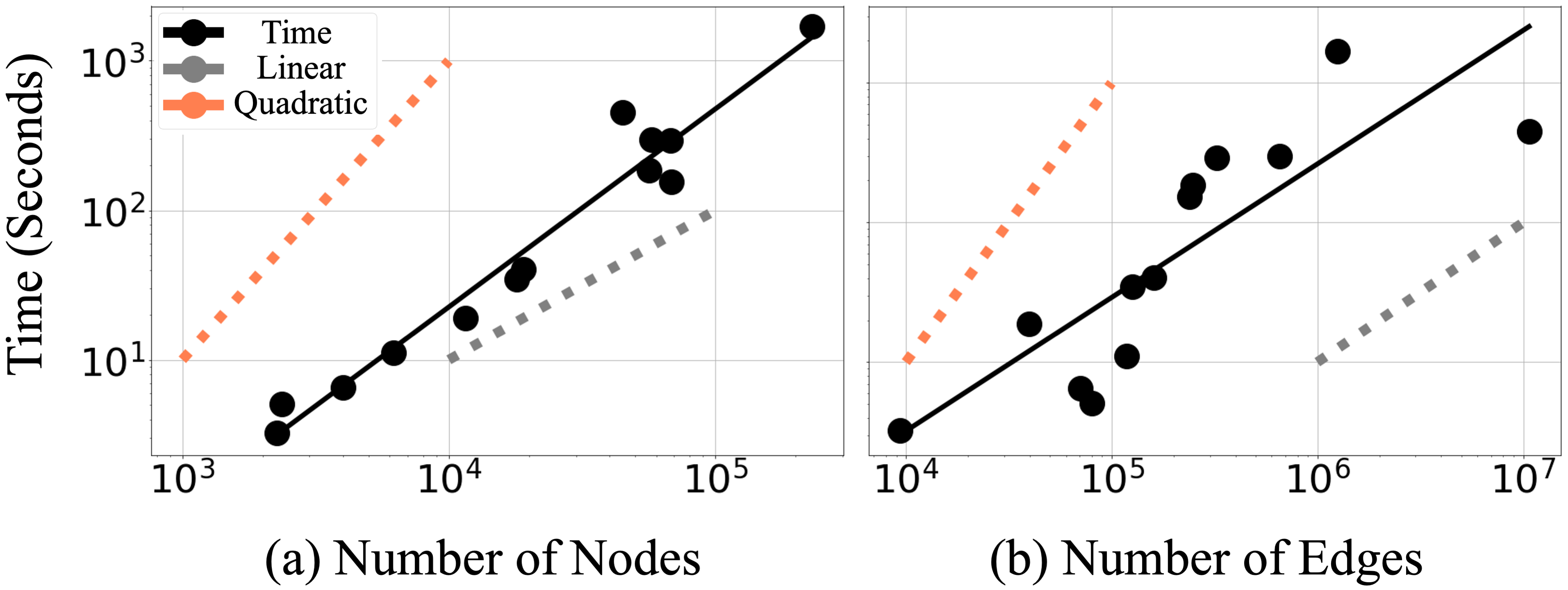}
    \caption{\method is sub-quadratic on the number of nodes (a) and linear on the number of edges (b).}
    \label{fig:runtime}
    \vspace{-0.6cm}
\end{figure}
\subsection{Parameters (\rThree)}
\label{subsec:params}
\noindent \textbf{Setup.} We evaluate the quality of different groupings (\S~\ref{subsec:groupings}), and how the number of groups in each affects performance. On four graphs, \yeast, \arxiv, \reddit, and \epinions, we run \method with groupings \groupingOne, \groupingTwo, and \groupingThree, varying the number of groups $|\grouping| \in \{5, 10, 25, 50, 75, 100\}$. We also investigate pairs of groupings, and the combination of all three groupings, via grid search of the number of groupings in each. We also evaluated $\toler \in \{1M, 10M, 25M, 50M\}$, the tolerance for searching equivalence classes exactly vs. approximately with LSH, and $\bailoutTol \in \{0, 0.1, 0.25, 1/3, 0.5, 2/3\}$, the fraction of training pairs we allow the proximity function to miss before we bailout of an equivalence class.

\vspace{0.1cm}
\noindent \textbf{Results.} The results for the individual groupings are shown in Fig.~\ref{fig:num-groups}. Grouping by log-binning nodes based on their degree, (i.e., \groupingOne) is in general the strongest grouping. Across all three groupings, we find that $|\grouping| = 25$ is a good number of groups. We found that using all three groupings was the best combination, with 25 log-bins, 5 structural clusters, and 5 communities (we omit the figures for brevity). For individual groupings, we observe diminishing returns, and in multiple groupings, slightly diminished performance when the number of groups in each grows large.  We omit the figures for $\toler$ and $\bailoutTol$, but found that $\toler = 25M$ and $\bailoutTol = 0.5$ were the best parameters. 

\begin{figure}[t]
\vspace{-0.25cm}
    \centering
    \includegraphics[width=\columnwidth]{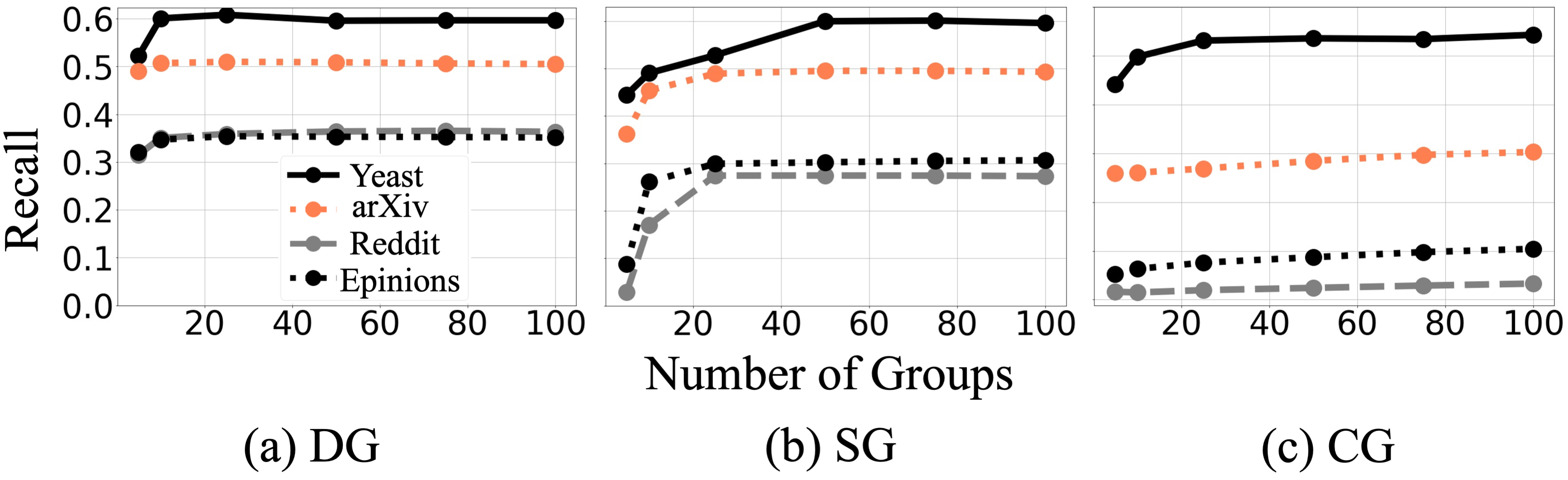}
    \caption{Number of groups for groupings \groupingOne, \groupingTwo, and \groupingThree}
    \label{fig:num-groups}
    \vspace{-0.6cm}
\end{figure}\section{Conclusion}
In this paper, we focus on the under-studied and challenging problem of identifying a moderately-sized set of node pairs
for a link prediction method to make decisions about. 
We mitigate the vastness of the search-space, filled with mostly non-links, by considering not just proximity, but also how much a pair of nodes \textit{resembles} observed links. We formalize this idea in the Future Link Location Model, show its theoretical connections to stochastic block models and proximity models, and introduce
an algorithm, \method, that leverages it to return high-recall candidate sets, with only a tiny fraction of all pairs. 
Via our resemblance insight, \method's strong performance generalizes from social networks to protein networks. Future directions include investigating the directionality of links, since the roadmap can incorporate this information, and extending to heterogeneous graphs with many edge and node types, like knowledge graphs.
\section*{Acknowledgements}
{\small
This work is supported by
an NSF GRF, NSF Grant No.\ IIS 1845491, Army Young Investigator Award No.\ W9-11NF1810397, 
and Adobe, Amazon, and Google faculty awards.
}

\balance
\bibliographystyle{plain}
\bibliography{abbrev,references}

\vspace{-0.1cm}
\appendix

\subsection{Effect of $\budget$} 
\label{subsec:effect-budget}
Table~\ref{table:vary-k} gives results (avgs over 3 seeds) for multiple values of $\budget$.
For reasonable values of $\budget$---roughly up to an order of magnitude greater than $\numEdges$---results are mostly stable. The main exception is for small $\budget$, where \ensemble performs well in some cases. This is consistent with its design: to return a small, accurate set of top-$\budget$ predictions, rather than a candidate set.

\vspace{-0.3cm}
\begin{table}[b!]
        \centering
    	\caption{We report ``$< \budget$'' if 
    	fewer than $\budget$ pairs are returned.
    	Black cells indicate values of $\budget$ outside the scale of the dataset.}
    	\vspace{-0.12cm}
    	\label{table:vary-k}
    	\resizebox{\columnwidth}{!}{
    	\begin{tabular}{lccc|ccc}
    		\toprule
    		& \multicolumn{3}{c}{\protein} & \multicolumn{3}{c}{\fbTwo} \\
    		\toprule
    		\emph{Metric} & \tiny{\ensemble} & \tiny{\adad} & \tiny{\methodMulti} & \tiny{\ensemble} & \tiny{\adad} & \tiny{\methodMulti} \\
    		\toprule
            R@10K & \textbf{0.2629} & 0.1909 & 0.2251 & 0.0093 & 0.0103 & \textbf{0.0172}\\
            \midrule
            R@100K & $< \budget$ & 0.8443 & \textbf{0.9035} & 0.0422 & 0.0672 & \textbf{0.1106}\\
            \midrule
            R@1M & $< \budget$ & 0.9747 & \textbf{0.9847} & 0.0970 & 0.2832 & \textbf{0.3781}\\
            \midrule
            R@5M & \cellcolor{black} N/A & \cellcolor{black} N/A & \cellcolor{black} N/A & $< \budget$ & 0.5561 & \textbf{0.5762}\\
      	    \bottomrule
    	\end{tabular}
    	}
\vspace{-0.1cm}
\end{table}

\begin{table}[b!]
    \centering
    \caption{Input Proximity Model for \lpm and \method, and parameter deviations from default for \ensemble.}
    \label{tab:params}
    \vspace{-0.12cm}
    \resizebox{\columnwidth}{!}{
    \begin{tabular}{lcccc}
         \toprule
         Graph & \lpm & \methodDeg & \methodMulti & \ensemble \\
         \midrule
         \yeast & \netmfTwo & \netmfTwo & \netmfTwo & $\epsilon = 0.5$ \\
         \dblp & \netmfTwo & \netmfTwo & \netmfTwo & Default\\
         \fbOne & \netmfOne & \adad & \adad & $\epsilon = 0.1$\\
         \movie & \bine & \bine & \bine & $\epsilon = 0.05$ \\
         \protein & \netmfTwo & \netmfTwo & \netmfTwo & $\epsilon = 0.75$ \\
         \arxiv & \netmfTwo & \adad & \adad & Default\\
         \mathOverflow & \netmfTwo & \netmfOne & \netmfOne & $\epsilon, \mu, f = 0.5, 0.3, 0.3$\\
         \enron & \netmfTwo & \netmfOne & \adad & Default\\
         \reddit & \netmfTwo & \adad & \adad & Default\\
         \epinions & \netmfOne & \adad & \adad & Default\\
         \fbTwo & \netmfTwo & \adad & \adad & Default\\
         \digg & \netmfTwo & \netmfOne & \netmfOne & Default\\
         \proteinAlt & \netmfOne & \netmfTwo & \netmfTwo & Default \\
         \bottomrule
    \end{tabular}
    }
\end{table}

\subsection{Complexity Analysis} Let $\complexityGroupings$ be the time complexity of the node grouping (\sOne). 
Computing the expected number of new edges in each cell (and variance), directly from the observed links, (\sTwo) is $\order(\numEdges)$. 
The complexity of searching equivalence classes (\sThree) comes from hashing each node in the decomposition $\order(\numHashesMax)$ times, and finding the $\target_i$ closest pairs in the $\order(\target_i)$ pairs that land in the same bucket: $\sum_{\memberPartEl \in \memberPart} \order(|\decompX \cup \decompY|\numHashesMax + \target_i) = \order(\numNodes \numHashesMax + \budget)$.
This assumes that we do not encounter unrealistic scenarios, e.g., the embeddings being equivalent and hence inseparable, and that $\numHashesMax$ is set large enough that the volume of tree leaves is not asymptotically larger than $\order(\target_i)$. 
Adding from the global pool (\sFour) takes $\order(\budget)$ time, since $|\globalPool| = \order(\budget)$ and can be maintained in sorted order (similar to the merge in merge sort). Thus, the total time complexity is $\order(\complexityGroupings + \numEdges + \numNodes \numHashesMax + \budget)$.

\label{sec:appendix}

\end{document}